\documentclass[11pt]{article}       
\usepackage{geometry}               
\usepackage{amsmath} % assumes amsmath package installed
\usepackage{amssymb}  % assumes amsmath package installed
\usepackage{amsthm}  % assumes amsmath package installed

\usepackage{mathtools}
\usepackage{amsthm}
\usepackage{amsfonts}
\usepackage{bbm}
\usepackage{hyperref}
\usepackage{cleveref}
\usepackage{algorithm}
\usepackage{algpseudocode}
\usepackage{cite}
\usepackage[dvipsnames]{xcolor}
\usepackage{bm}

\usepackage{booktabs}
\usepackage{multirow}
\usepackage{dsfont}
\usepackage{xspace}

\newtheorem{theorem}{Theorem}
\newtheorem{lemma}{Lemma}

\newtheorem{assumption}{Assumption}
\newtheorem{proposition}{Proposition}
\newtheorem{remark}{Remark}
\newtheorem*{problem}{Problem Statement}

\Crefname{assumption}{Assumption}{Assumptions}
\Crefname{theorem}{Theorem}{Theorems}
\Crefname{lemma}{Lemma}{Lemmas}
\Crefname{figure}{Fig.}{Figs.}

\newcommand{\ellone}{$\mathcal{L}_1$\xspace}
\newcommand{\ellonedrac}{$\mathcal{L}_1$-DRAC\xspace}
\newcommand{\xdisc}{\bar{x}}
\newcommand{\udisc}{\bar{u}}
\newcommand{\wdisc}{\bar{w}}
\newcommand{\ydisc}{\bar{y}}
\newcommand{\zdisc}{\bar{z}}
\newcommand{\xdisctrue}{x}

\newcommand{\xdisclaw}[2]{\bar{\mathbb{X}}^d_{#1|#2}}
\newcommand{\zdisclaw}[2]{\bar{\mathbb{Z}}^d_{#1|#2}}
\newcommand{\ydisclaw}[2]{\bar{\mathbb{Y}}^d_{#1|#2}}

\newcommand{\Wass}{\mathcal{W}}
\newcommand{\norm}[1]{\left\| #1 \right\|}

\newcommand{\frameworkname}{DRP-$\mathcal{L}_1$AC}

\newcommand{\Boldomega}{\omega}
\newcommand{\Lhat}[1]{\hat{\Lambda}\left(#1\right)}
\newcommand{\BoldTs}{T_s}
\newcommand{\br}[1]{\left( #1 \right)}
\newcommand{\cbr}[1]{\left\{ #1 \right\}}
\newcommand{\expo}[1]{e^{#1}}
\newcommand{\indicator}[2]{\mathbf{1}_{#1}\left(#2\right)} 
\newcommand{\ULt}[1]{U_{\mathcal{L}_1,#1}}
\newcommand{\Xtildet}[1]{\tilde{X}_{#1}}
\newcommand{\Xhatt}[1]{\hat{X}_{#1}}
\newcommand{\Xt}[1]{X_{#1}}

\title{Distributionally Robust Planning with $\mathcal{L}_1$ Adaptive Control\thanks{
This work was supported in part by the Higher Education and Science Committee of RA (Research project 24FP-C017), Air Force Office of Scientific Research (AFOSR) Grant FA9550-25-1-0274, and by the National Science Foundation (NSF) under Grants 2135925, 2331878, 2311085, 2502857, and 2515359.
}
} % 
\author{Astghik Hakobyan, Amaras Nazarians, Aditya Gahlawat,  \\ Naira Hovakimyan, and Ilya Kolmanovsky\thanks{A. Hakobyan is with the Center for Scientific Innovation and Education and National Polytechnic University of Armenia, Yerevan, Armenia
        {\tt\small astghik.hakobyan@csie.am}. 
A. Gahlawat and N. Hovakimyan are with the Department of Mechanical Science and Engineering, Grainger College of Engineering, University of Illinois at Urbana-Champaign, Urbana, IL, USA
        {\tt\small \{gahlawat, nhovakim\}@illinois.edu}. 
A. Nazarians is with the Center for Scientific Innovation and Education and American University of Armenia, Yerevan, Armenia
        {\tt\small amaras\_nazarians@edu.aua.am}. 
I. Kolmanovsky is with the Department of Aerospace Engineering, University of Michigan, Ann Arbor, MI, USA.
{\tt\small ilya@umich.edu}.
}
}

\date{}

\begin{document}

\maketitle

\begin{abstract}
Safe operation of autonomous systems requires robustness to both model uncertainty and uncertainty in the environment. We propose \frameworkname,  a hierarchical framework for stochastic nonlinear systems that integrates distributionally robust model predictive control (DR-MPC) with \ellone-adaptive control. The key idea is to use the \ellone-adaptive controller's online distributional certificates that bound the Wasserstein distance between nominal and true state distributions, thereby certifying the ambiguity sets used for planning without requiring distribution samples. Environmental uncertainty is captured via data-driven ambiguity sets constructed from finite samples. These are incorporated into a DR-MPC planner enforcing distributionally robust chance constraints over a receding horizon. Using Wasserstein duality, the resulting problem admits tractable reformulations and a sample-based implementation. 
We show theoretically and via numerical experimentation that our framework ensures certifiable safety in the presence of \emph{simultaneous system and environmental uncertainties.}
\end{abstract}

\section{Introduction}
Certifiably safe operation of autonomous systems under uncertainty remains a fundamental challenge in control, particularly when both system dynamics and the environment are uncertain.
Real-world systems are subject to both \emph{epistemic} uncertainty from imperfect system and environmental knowledge and \emph{aleatoric} uncertainty due to exogenous stochastic disturbances. 
Consequently, system behavior can deviate significantly from nominal predictions, making safety guarantees difficult, if not infeasible.

Existing approaches face a fundamental trade-off between robustness and performance.
Classical robust and adaptive control methods~\cite{zhou1998essentials,herbert2017fastrack,singh2019robust} provide stability under bounded uncertainty but rely on worst-case analysis, often leading to conservative behavior. 
In contrast, stochastic approaches enforce safety in an average-case or high-probability sense by reasoning over distributions rather than uncertainty sets.
Such representations allow control synthesis at the level of state distributions and are amenable to data-driven approaches~\cite{achiam2017constrained,liu2022constrained}.
However, despite strong empirical performance, they lack robustness guarantees and are sensitive to distribution shifts, limiting their use in safety-critical systems.

%%%%%%%%%%%%%%%%%%%%%%%%%%%%%%%%%%%%%%%%%%%%%%%%%%%%%%%
\begin{figure}
    \centering
    \includegraphics[width=0.8\linewidth]{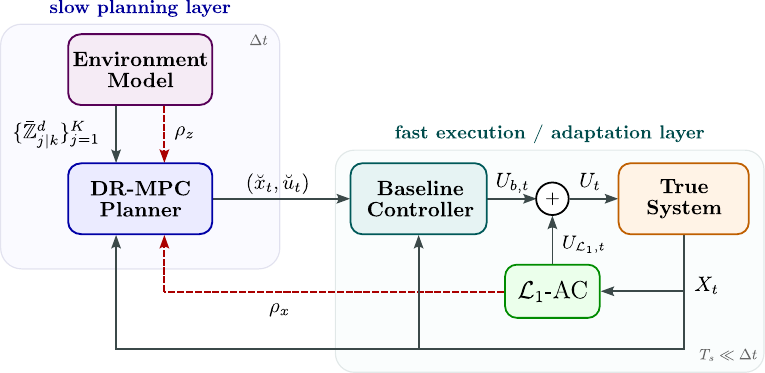}
    \caption{Overview of the proposed \frameworkname \ framework. A DR-MPC planner generates reference trajectories satisfying safety constraints, which are tracked by a baseline controller augmented with an \ellone-AC term. The \ellone-AC certifies that the true state distribution remains within a prescribed Wasserstein tube around the nominal distribution.
    }
    \label{fig:overview}
\end{figure}
%%%%%%%%%%%%%%%%%%%%%%%%%%%%%%%%%%%%%%%%%%%%%%%%%%%%%%

The distributional perspective, which operates directly over probability measures, provides a promising alternative. Ambiguity sets that capture the effects of uncertainties enable safety through chance or conditional value-at-risk (CVaR) constraints~\cite{zymler2013distributionally, xie2021distributionally, ji2021data}. 
This perspective has led to distributionally robust control methods~\cite{kim2023distributional, taskesen2023distributionally, queeney2023risk}, with distributionally robust model predictive control (DR-MPC)~\cite{coulson2021distributionally, mcallister2024distributionally, hakobyan2021wasserstein, zolanvari2025iterative} providing a natural receding-horizon formulation. 

However, a key limitation of DR-MPC is the assumption of the \emph{existence and knowledge of uncertainty-induced ambiguity sets}.
While environmental uncertainty can often be characterized from multiple data samples (e.g., via onboard sensors), sampling from the system's state distributions is infeasible.
At each time step, the controller has access to the \textbf{feedback signal, which is a single realization of a time-varying distribution}.
As a result, there is no mechanism to \emph{certify that the true (uncertain) closed-loop system remains within the ambiguity set used for planning}, and safety guarantees may fail under uncertain system dynamics. 

To address this limitation, we leverage \emph{\ellone-adaptive control} (\ellone-AC), a robust adaptive controller that guarantees uniformly bounded state deviations under epistemic uncertainty and has been validated on systems such as NASA’s AirStar and Calspan’s Learjet~\cite{gregory2010flight,ackerman2017evaluation}.
Its robustness properties have been exploited in data-driven settings, including reinforcement learning~\cite{sung2024robust}, safe learning-based control~\cite{gahlawat2020l1,pmlr-v144-gahlawat21a}, and model predictive control (MPC) for trajectory-level robustness~\cite{pereida2021robust,pravitra2020L1}. 
Importantly, the recently developed \ellone-distributionally robust adaptive control (\ellonedrac)~\cite{gahlawat2025DRAC} establishes finite-time Wasserstein bounds between nominal and true state distributions, defining Wasserstein ambiguity tubes used in covariance steering~\cite{gahlawat2025wasserstein}.
These methods, including DR-MPC, either perform \emph{planning without certifiable guarantees}, or \emph{fail to ensure safety under simultaneous model and environmental uncertainty}. 

We propose \textbf{\frameworkname}, a novel framework, illustrated in Fig.~\ref{fig:overview}, that \emph{closes the loop between the high-level planning and low-level controller through distributional certificates}. Our approach enables \textbf{safe operation under both model and environmental uncertainties} by certifying safety over joint system--environment distributions.
We integrate a DR-MPC planner with an \ellone-adaptive controller in a bidirectional architecture: the planner generates reference trajectories, while the low-level \ellone-adaptive controller tracks the reference and communicates \emph{sample-free Wasserstein certificates} that bound the deviation between nominal and true state distributions, thereby certifying the system-side ambiguity set.
Environmental uncertainty is captured via data-driven ambiguity sets, and the resulting DR-MPC enforces distributionally robust chance constraints (DR-CCs).
We show that feasibility of the DR-MPC problem, together with the \ellonedrac certificate, \emph{guarantees stagewise safety of the true closed-loop system}.

The main contributions of this paper are as follows:
\begin{itemize}
    \item A hierarchical DR-MPC framework that couples planning and adaptation via a joint Wasserstein ambiguity set, where the \emph{a priori} known system ambiguity radius is certified online using \ellone-AC without requiring distribution samples, and the environmental radius is obtained via finite-sample concentration bounds;
    \item A closed-loop distributional certification mechanism that enables \emph{provable stagewise safety guarantees} under simultaneous model and environmental uncertainty; and
    \item A tractable and geometry-agnostic reformulation of DR-CCs using CVaR and Wasserstein duality, enabling a sample-based DR-MPC implementation.
\end{itemize}
We validate the proposed framework through numerical experiments demonstrating certified safe operation under both model mismatch and environmental uncertainty.

The paper is organized as follows. \Cref{sec:prelim} introduces the system and the environment, as well as the problem statement. \Cref{sec:dr_l1} presents the proposed framework. \Cref{sec:tract_mpc} develops a tractable DR-MPC reformulation, and \Cref{sec:exp} provides numerical results.

\medskip

\noindent \textbf{\emph{Notation.}} We denote the set of nonnegative reals and positive integers by $\mathbb{R}_{\ge 0}$ and
$\mathbb{N}$.
For any $x \in \mathbb{R}^n$, $\|x\|$ denotes the Euclidean norm.
We denote by $\mathcal{P}_q(\mathcal{Y})$ the set of probability measures
on Polish space $\mathcal{Y}$ with finite $q$-th moments.
The law of a random variable $X$ is denoted by $\mathfrak{L}(X)$, and for probability measures $\mu$ and $\nu$, $\mu \otimes \nu$ denotes their product measure. 
All stochastic processes in the manuscript are defined on a filtered probability space $(\Omega,\mathcal F,\{\mathcal F_t\}_{t\ge0},\mathbb P)$ satisfying the usual conditions.
For two probability measures $\mu,\nu \in \mathcal P_q (\mathcal Y)$, defined on a Polish space $(\mathcal Y,d_{\mathcal Y})$, the Wasserstein distance of order $q\ge1$ is defined as 
$
\Wass_q(\mu,\nu)
=
\bigl(
\inf_{\pi\in\Pi(\mu,\nu)}
\int_{\mathcal Y\times\mathcal Y}
d_{\mathcal Y}(y,\bar y)^q
\,d\pi(y,\bar y)
\bigr)^{1/q}$,
where $\Pi(\mu,\nu)$ denotes the set of couplings with marginals $\mu$ and $\nu$.
The corresponding Wasserstein ambiguity set centered at $\nu \in \mathcal{P}_q(\mathcal{Y})$ is defined as
\begin{equation}\label{eqn:amb_set}
    \mathbb{B}_q(\nu,\zeta)
    :=
    \left\{
    \mu \in \mathcal{P}_q(\mathcal{Y})
    \;\middle|\;
    \Wass_q(\mu,\nu) \le \zeta
    \right\},
\end{equation}
where $\zeta \in \mathbb{R}_{>0}$ denotes the ambiguity radius.

\section{Preliminaries and Problem Statement}\label{sec:prelim}

We now describe the uncertain system we wish to control and the uncertain environment in which the system operates.  Then, in~\Cref{subsec:probstat} we introduce the problem we address. 

\subsection{Uncertain System and Environmental Dynamics}

We consider an \textbf{uncertain process} $X_t \in \mathbb{R}^n$ evolving according to the following nonlinear It\^o stochastic differential equation:
\begin{equation}\label{eqn:true_sys}
    dX_t
    =
    \left(A_\mu X_t + B U_t + B \Lambda_\mu(t, X_t)\right)\,dt 
    +
    \left(A_\sigma+B \Lambda_\sigma(t,X_t)\right)\,dW_t,
    \quad
    X_0 \sim \xi_0,
\end{equation}
where $U_t \in \mathbb{R}^m$ is an input process to be specified later, and $A_\mu\in\mathbb R^{n\times n}$, $A_\sigma\in\mathbb R^{n\times d}$, and $B\in\mathbb R^{n\times m}$ are \textbf{known}.
The input operator $B$ is assumed to be full rank.
The \textbf{unknown} functions $\Lambda_\mu:\mathbb R_{\ge0}\times\mathbb R^n\to\mathbb R^m$ and $\Lambda_\sigma:\mathbb R_{\ge0}\times\mathbb R^n\to\mathbb R^{m\times d}$ represent the drift and diffusion uncertainties, respectively.  
Moreover, $W_t \in \mathbb{R}^d$ is an $\mathcal{F}_t$-adapted Brownian motion and $\xi_0$ is the initial distribution independent of $\cup_{t\geq0}W_t$. 

Dropping the unknown functions from~\eqref{eqn:true_sys} produces the nominal (uncertainty-free) dynamics
\begin{equation}\label{eqn:nom_sys}
d\bar{X}_t
=
(A_\mu \bar{X}_t + B \bar{U}_t)\,dt
+
A_\sigma\,d\bar{W}_t,
\;
\bar{X}_0 \sim \bar{\xi}_0,
\end{equation}
where $\bar{X}_t \in \mathbb{R}^n$ is the \textbf{nominal state process} and $\bar W_t\in\mathbb R^d$ is another $\mathcal{F}_t$-adapted Brownian motion that is independent of $W_t$ driving~\eqref{eqn:true_sys}.
The initial distribution $\bar\xi_0$ is independent of $\cup_{t\geq0} \bar{W}_t$. 
We denote the \textbf{uncertain and nominal (known) probability laws} by $\mathbb X_t := \mathfrak{L}(X_t)$ and $\bar{\mathbb X}_t := \mathfrak{L}(\bar X_t)$, respectively.
%%%%%%%%%%%%%%%%%%%%%%%%%%%%%%%%%%%%%%%%%%%%%%%
\begin{assumption}[Uncertainty regularity]\label{ass:true_sys}
The uncertainties $\Lambda_\mu(t,\xi)$ and $\Lambda_\sigma(t,\xi)$ are globally Lipschitz over $(t,\xi) \in \mathbb{R}_{\ge0} \times \mathbb{R}^n$. 
The diffusion uncertainty is additionally growth bounded as $\norm{\Lambda_\sigma(t,\xi)}^2 \leq C \left(1 + \norm{\xi}^2\right)^{\frac{1}{2}}$, $\forall (t,\xi) \in \mathbb{R}_{\ge0} \times \mathbb{R}^n$, where $C$ is some positive scalar.
\end{assumption}
% \begin{assumption}[Uncertainty regularity]\label{ass:true_sys}
% The functions $\Lambda_\mu$ and $\Lambda_\sigma$ satisfy the growth bounds
% $\|\Lambda_\mu(t,x)\|^2 \le \Delta_{\mu_1}^2+\Delta_{\mu_2}^2\|x\|^2$ and
% $\|\Lambda_\sigma(t,x)\|_F^2 \le (\Delta_{\sigma_1}^2+\Delta_{\sigma_2}^2\|x\|^2)^{1/2}$ for all $(t,x)\in\mathbb R_{\ge0}\times\mathbb R^n$, where $\Delta_{\mu_1}, \Delta_{\mu_2}, \Delta_{\sigma_1},\Delta_{\sigma_2} > 0$ are known constants.
% Moreover, $\Lambda_\mu$ is globally Lipschitz in $(t,x)$, i.e.,
% $\|\Lambda_\mu(t,x)-\Lambda_\mu(t',x')\|\le\hat L_\mu|t-t'|+L_\mu\|x-x'\|$,
% while $\Lambda_\sigma$ is Lipschitz in time and H\"older continuous with exponent $1/2$ in the state variable, i.e.,
% $\|\Lambda_\sigma(t,x)-\Lambda_\sigma(t',x')\|_F\le\hat L_\sigma|t-t'|+L_\sigma\|x-x'\|^{1/2}$.
% \end{assumption}
%%%%%%%%%%%%%%%%%%%%%%%%%%%%%%%%%%%%%%%%%%%%%%%%%%
%%%%%%%%%%%%%%%%%%%%%%%%%%%%%%%%%%%%%%%%%%%%%%%%%%
\begin{remark}
    The global Lipschitz continuity of $\Lambda_\mu$ in \Cref{ass:true_sys} further implies its linear growth over $\mathbb{R}^n$.
    We place a stricter sublinear growth condition on $\Lambda_\sigma$, which, for a diffusion uncertainty, is more general compared to existing results in the literature where drift uncertainties are assumed to be uniformly bounded.    
\end{remark}
%%%%%%%%%%%%%%%%%%%%%%%%%%%%%%%%%%%%%%%%%%%%%%%%%%
 
Under \Cref{ass:true_sys} and admissible inputs $U_t$, \eqref{eqn:true_sys} admits a unique strong solution on $[0,T]$, $T<\infty$~\cite[Thm.~2.3.1]{mao2007stochastic}. 
The nominal system~\eqref{eqn:nom_sys} is trivially well-posed, since $\bar{U}_t$ will be designed as a linear feedback input.

We model the \textbf{uncertain environment} by a random and temporally evolving process $Z_t \in \mathbb{R}^{n_z}$, which represents the environment state.
We define unsafe regions (occupied by obstacles) by the set-valued map $\mathcal{O} : \mathbb{R}^{n_z} \rightrightarrows \mathbb{R}^n$, where $\mathcal{O}(z) \subset \mathbb{R}^n$ is closed, for all $z \in \mathbb{R}^{n_z}$.
Thus, at any $t \in \mathbb{R}_{\geq 0}$, the \emph{random} and \emph{uncertain} unsafe region is given by $\mathcal{O}(Z_t)$.   
We assume the obstacles within the environment do not react to the system, i.e., the process $\{Z_t\}$ is independent of $\{X_t\}$. Since the environment is uncertain, its law $\mathbb{Z}_t := \mathfrak{L}(Z_t)$ is unknown. However, real-time measurements from on-board sensors provide finite samples from $\mathbb{Z}_t$, which we use to construct a nominal environmental distribution $\bar{\mathbb{Z}}_t$.%
\footnote{We provide the details on the construction of $\bar{\mathbb{Z}}_t$ later in the manuscript.}
% We consider the case where the obstacles within the environment do not react to the system, i.e. the environment process $\{Z_t\}$ is independent of the system $\{X_t\}$. However, since a system can obtain real-time samples of the environment it is operating in using on-board sensors, we assume access to a finite number of samples from $\mathbb{Z}_t$ which allows us to construct a nominal environment distribution $\bar{\mathbb{Z}}_t$.

% is assumed to be exogenous and independent of the
% system state process $\{X_t\}$. Its probability law $\mathbb{Z}_t := \mathfrak{L}(Z_t)$ is unknown in practice, and we \annotate{instead assume access to a nominal distribution $\bar{\mathbb{Z}}_t$, obtained from historical observations or samples generated by an environment prediction model}{Not real-time?}.

\subsection{Problem Statement}\label{subsec:probstat}

To ensure the safe operation of the uncertain process $X_t$ evolving within an uncertain environment $Z_t$, we introduce \emph{lifted (concatenated) representation}  $Y_t := (X_t,Z_t) \in \mathcal{Y} := \mathbb{R}^n \times \mathbb{R}^{n_z}$.
Then, we say that the \textbf{system operates safely} with probability $\mathbf{1 - }\boldsymbol{\beta}$, $\boldsymbol{\beta}\, \mathbf{\in (0,1)}$, if
    \begin{equation}
        \mathbb{P}\left(Y_t \in \mathcal{C}\right) \leq \beta, \quad \forall t \geq 0,
        \label{eqn:chance_const_lifted}
    \end{equation}    
where $\mathcal{C} :=  \left\{ (x,z) \in \mathbb{R}^n \times \mathbb{R}^{n_z} : x \in \mathcal{O}(z) \right\}$ is the \emph{lifted unsafe set}, capturing all joint configurations $(x,z)$ for which the system state $x$ lies within the unsafe region $\mathcal{O}(z)$ induced by the environment state $z$.

The \emph{uncertain lifted law} $\mathbb Y_t = \mathbb X_t \otimes \mathbb Z_t$ is unknown due to uncertainty in both the system and environment.
Instead, we only have access to the \emph{nominal lifted law} $\bar{\mathbb Y}_t = \bar{\mathbb X}_t \otimes \bar{\mathbb Z}_t$ of the \emph{nominal lifted process} $\bar{Y}_t := (\bar{X}_t,\bar{Z}_t)$.\footnote{The lifted laws of the respective lifted processes are their product measure due to the assumed independence of the system and the environment.}
The mismatch between the uncertain (true) $\mathbb Y_t$ and its nominal (known) estimate $\bar{\mathbb Y}_t$ motivates the problem we aim to address.
\begin{problem}
Using only the nominal law $\bar{\mathbb Y}_t$, we wish to design the input process $U_t$ in~\eqref{eqn:true_sys} such that the unknown law $\mathbb Y_t$ of the uncertain lifted process $Y_t$ satisfies~\eqref{eqn:chance_const_lifted} for a prescribed probability $\beta \in (0,1)$, which implies that the unknown process $X_t$ operates safely in the uncertain environment $Z_t$ with probability $1-\beta$.
\end{problem}
% Consequently, direct evaluation of~\eqref{eqn:chance_const_lifted}
% is infeasible.

% \vspace{1cm}
% \todo[inline, backgroundcolor=BlueGreen, textcolor=white]{Cannibalizing the following for the problem statement.}
% To describe safety in a unified manner, we introduce the lifted random
% variable $Y_t := (X_t,Z_t) \in \mathcal{Y}
% := \mathbb{R}^n \times \mathbb{R}^{n_z}$,

% and define the lifted unsafe set $\mathcal{C} :=
% \{(x,z) \in \mathbb{R}^n \times \mathbb{R}^{n_z} : x \in \mathcal{O}(z)\}$.
% The unsafe event can then be written equivalently as $Y_t \in \mathcal{C}$.
% Accordingly, the safety requirement can be expressed through the chance
% constraint
% \begin{equation}
% \label{eqn:chance_const_lifted}
% \mathbb{P}\left(Y_t \in \mathcal{C}\right)
% \leq \beta,
% \end{equation}
% where $\beta \in (0,1)$ denotes the allowable violation probability.

% Since both the system and environment states are uncertain and their true
% probability laws are unknown, the probability law of $Y_t$ is generally
% unavailable. Consequently, direct evaluation of~\eqref{eqn:chance_const_lifted}
% is intractable.

% \vspace{1cm}
% \todo[inline]{The following absorbed by the subsequent section, need confirmation!}

%%%%%%%%%%%%%%%%%%%%%%%%%%%%%%%%%%%%%%%%%%%%%%%%%%%%%%%%%%
\section{Design of the \frameworkname \ Framework}\label{sec:dr_l1}

The key difficulty in solving the problem above lies in the absence of a quantifiable relationship between the true lifted law $\mathbb{Y}_t$ and its nominal counterpart $\bar{\mathbb{Y}}_t$. Without such a measure of deviation, the chance constraint~\eqref{eqn:chance_const_lifted} cannot be enforced using nominal information alone.

We address this by introducing a distributional bound: suppose there exists a known $\rho_y \in \mathbb{R}_{>0}$ such that $\Wass_{2p}(\mathbb{Y}_t,\bar{\mathbb{Y}}_t) \leq \rho_y$ for all $t \geq 0$ and some $p \in \mathbb{N}$. 
Then, safety can be enforced via the following DR-CC:
\begin{equation}\label{eqn:dr_chance_prop}
    \sup_{\nu \in \mathbb{B}_{2p}(\bar{\mathbb{Y}}_t, \rho_y)}
    \mathbb{P}_{\nu}(Y_t \in \mathcal{C}) \le \beta, \quad \forall t \geq 0,
\end{equation}
where $\mathbb{B}_{2p}(\bar{\mathbb{Y}}_t, \rho_y)$ denotes the $2p$-Wasserstein ambiguity set of radius $\rho_y$ centered at $\bar{\mathbb{Y}}_t$ as defined in~\eqref{eqn:amb_set}, and $\mathbb{P}_\nu$ denotes probability under the measure $\nu$, i.e., $\mathbb{P}_{\nu}(Y_t \in \mathcal{C}) := \nu(\mathcal{C})$.

We therefore design \frameworkname\ to construct such an ambiguity set, so that the DR-CC~\eqref{eqn:dr_chance_prop} can be satisfied in a tractable way. 
The resulting framework, illustrated in~\Cref{fig:overview}, consists of two hierarchical components:
\begin{enumerate}
    \item a high-level \textbf{DR-MPC planner} that enforces~\eqref{eqn:dr_chance_prop}, assuming an ambiguity set containing $\mathbb{Y}_t$ is available, i.e., $\mathbb{Y}_t \in \mathbb{B}_{2p}(\bar{\mathbb{Y}}_t, \rho_y)$; and
    \item a low-level \textbf{robust-adaptive tracking controller}, based on the \ellone-AC architecture, that guarantees $\Wass_{2p}(\mathbb{X}_t,\bar{\mathbb{X}}_t) \leq \rho_x$, where $\rho_x$ is known \emph{a priori}. 
    Combined with the data-driven bound $\Wass_{2p}(\mathbb{Z}_t,\bar{\mathbb{Z}}_t) \leq \rho_z$ for the environment, this yields $\rho_y = \rho_x + \rho_z$.
\end{enumerate}

We next present the design of the high-level DR-MPC planner, followed by the low-level tracking controller.
%%%%%%%%%%%%%%%%%%%%%%%%%%%%%%%%%%%%
\subsection{DR-MPC for High-Level Planning}\label{sec:drmpc}
% \todo[inline]{Turn $\bar{\mathbf{X}}$, $\bar{\mathbf{U}}$ into $\bar{\mathbf{x}}$, $\bar{\mathbf{u}}$}
% \todo[inline]{Mention on the discrete nature of MPC and the constraints.}

As stated above, we design a high-level DR-MPC for the nominal system such that the DR-CC~\eqref{eqn:dr_chance_prop} is enforced. As is standard in MPC, we consider a discrete-time formulation of the nominal system~\eqref{eqn:nom_sys} with sampling period $\Delta t > 0$.

\begin{remark}
The MPC layer is discrete-time, while the tracking controller is continuous-time. 
Hence, the DR-CC~\eqref{eqn:dr_chance_prop} is enforced only at sampling instants and does not guarantee inter-sample safety. 
This can be addressed by absorbing discretization error (as a function of $\Delta t$) into the ambiguity radius $\rho_y$, which we omit for simplicity.
\end{remark}

%%%%%%%%%%%%%%%%%%%%%%
Define $t_k := t_0 + k\Delta t$ and let $\xdisc_k \coloneqq \bar X_{t_k}$ and $\udisc_k \coloneqq \bar U_{t_k}$ denote the sampled nominal state and input, respectively. 
Under a zero-order-hold input, i.e., $\bar U_t = \udisc_k$ for $t \in [t_k,t_{k+1})$, the dynamics of $\xdisc_k$ evolve as
\begin{align}\label{eqn:nom_sys:discrete}
    \xdisc_{k+1}
    =
    A_d \xdisc_k + B_d \udisc_k + \wdisc_k,
    \quad 
    \xdisc_0 = \bar{X}_0 \sim \bar{\xi}_0,
\end{align}
where $\wdisc_k \overset{i.i.d.}{\sim} \mathcal{N}(0,\Sigma_d)$. 
The matrices $\Sigma_d \in \mathbb{S}^n$, $A_d \in \mathbb{R}^{n \times n}$ and $B_d \in \mathbb{R}^{n \times m}$ are obtained from the analytic solution of~\eqref{eqn:nom_sys} under piecewise constant input $\udisc_k$~\cite[Chp.~5]{oksendal2013stochastic}.

We denote by $\left\{t_{j|k}\right\}_{j=0}^{K-1}$ the prediction horizon of length $K \in \mathbb{N}$ with $t_{0|k} \coloneqq t_k$.
Over this horizon, we define the control sequence $\mathbf{\udisc}_{k} = \{\udisc_{j|k}\}_{j=0}^{K-1}$ and state sequence $\mathbf{\xdisc}_{k} = \{\xdisc_{j|k}\}_{j=0}^K$, with corresponding probability laws $\left\{ \xdisclaw{j}{k} \right\}_{j=0}^{K}$. 
At each step $j$, we further define the discrete nominal lifted law
 $\ydisclaw{j}{k} \coloneqq 
\xdisclaw{j}{k} \otimes \zdisclaw{j}{k}$ with $\zdisclaw{j}{k} = \mathfrak{L}(\zdisc_{j+k})$.

For a given failure probability $\beta \in (0,1)$ and a radius $\rho_y$ (\Cref{subsec:DistCerts}), we formulate the DR-MPC problem as follows:
\begin{subequations}\label{eqn:dr_mpc}
    \begin{align}
        \min_{\bar{\mathbf{U}}_{k}, \bar{\mathbf{X}}_{k}}\;
        & \sum_{j=0}^{K-1}\mathbb{E} \left[\ell_j(\xdisc_{j|k}, \udisc_{j|k})\right]
        +
        \mathbb{E} \left[\ell_N(\xdisc_{K|k})\right]
        \label{eqn:drmpc_cost}
        \\
        \text{s.t.}\;
        & \xdisc_{j+1|k}  = 
        A_d \xdisc_{j|k}  +  B_d \udisc_{j|k}  +  \wdisc_{j|k}, \qquad j = 0,\dots, K-1,
        \label{eqn:drmpc_dyn}
        \\
        & \udisc_{j|k}\in \mathcal{U}, \qquad j=0,\dots,K-1,
        \label{eqn:drmpc_set_1}
        \\
        & \mathbb{E}[\xdisc_{j|k}]\in\mathcal{X},
        \qquad j=1,\dots,K,
        \label{eqn:drmpc_set_2}
        \\
        & \xdisc_{0|k}= \bm{x},
        \label{eqn:drmpc_init}
        \\
        & \sup_{ \nu \in \mathbb{B}_{2p}(\ydisclaw{j}{k},\rho_y)}\!
        \mathbb{P}_{\nu}\big(\ydisc_{j|k} \! \in \mathcal C\big) \!
        \le \beta, \qquad j=1,\dots,K,
        \label{eqn:drmpc_safe}
    \end{align}
\end{subequations}
where $\ell_j:\mathbb{R}^n\times\mathbb{R}^m\to\mathbb{R}$ and $\ell_N:\mathbb{R}^n\to\mathbb{R}$ denote the stage and terminal costs, respectively, and $\bm{x} \in \mathbb{R}^n$ in~\eqref{eqn:drmpc_init} is the initial condition.\footnote{Throughout the paper, we adopt time-varying quadratic tracking costs of the form
$\ell_j(x,u) = (x-x^{\mathrm{des}}_j)^\top Q (x-x^{\mathrm{des}}_j)
+ (u-u^{\mathrm{des}}_j)^\top R (u-u^{\mathrm{des}}_j)$,
with terminal cost
$\ell_N(x) = (x-x^{\mathrm{des}}_N)^\top Q_f (x-x^{\mathrm{des}}_N)$,
where $x_j^{\mathrm{des}}$ and $u_j^{\mathrm{des}}$ denote the desired state and input at time step $j$, respectively, and $Q \succeq 0$, $R \succ 0$, and $Q_f \succeq 0$ are weighting matrices.}
Constraint~\eqref{eqn:drmpc_dyn} enforces the nominal dynamics, where $\{\wdisc_{j|k}\}$ is an i.i.d.\ Gaussian sequence with distribution $\mathcal N(0,\Sigma_d)$. Since $\xdisc_{j|k}$ is stochastic, pathwise constraints are generally intractable or overly conservative; we therefore impose the mean constraint~\eqref{eqn:drmpc_set_2}, which represents a standard relaxation in stochastic MPC, while~\eqref{eqn:drmpc_set_1} ensures input feasibility.
Finally,~\eqref{eqn:drmpc_safe} enforces the DR-CC~\eqref{eqn:dr_chance_prop}.

For each initial state $\bm{x}$, let $\mathbf{\udisc}_k^\star(x)
\coloneqq
\{\udisc_{j|k}^\star(x)\}_{j=0}^{K-1}$
denote an optimal control sequence of~\eqref{eqn:dr_mpc}, with associated state sequence
$\mathbf{\xdisc}_k^\star(x)
\coloneqq
\{\xdisc_{j|k}^\star(x)\}_{j=0}^{K}$. The induced MPC law is given by the first input,
$\kappa_{\mathrm{MPC}}(x) \coloneqq \udisc_{0|k}^\star(x)$.
Thus, although~\eqref{eqn:dr_mpc} is an open-loop finite-horizon problem, its receding-horizon implementation induces an implicit state-feedback law for~\eqref{eqn:nom_sys:discrete}.

At each planning step $k$, we initialize~\eqref{eqn:drmpc_init} with the measured (true) state $\bm{x} = X_{t_k}$ and compute the optimal nominal state-input trajectory
$\{(\xdisc_{j|k}^\star(\xdisctrue_k),\udisc_{j|k}^\star(\xdisctrue_k))\}_{j=0}^{K-1}$, which is then passed to a low-level tracking controller.

\subsection{Low-Level Robust Adaptive Tracking Control}

We now design the low-level robust adaptive control input $U_t$ for the uncertain system~\eqref{eqn:true_sys} to track the reference trajectories $(\mathbf{\xdisc}_k^\star(x), \mathbf{\udisc}_k^\star(x))$ generated by the DR-MPC. We decompose the input as
\begin{align}\label{eqn:aug_cont}
    U_t = U_{b,t} + U_{\mathcal{L}_1,t},
\end{align}
where $U_{b,t}$ is a \emph{baseline} controller designed for tracking in the absence of uncertainties, and $U_{\mathcal{L}_1,t}$ is the \ellone-AC augmentation designed to handle the epsitemic uncertainties in~\eqref{eqn:nom_sys}. 
The baseline controller $\mathcal{F}_{b}(t,\cdot):\mathbb{R}^n \rightarrow \mathbb{R}^m$, $t \geq 0$, is chosen to be a standard feedback-feedforward operator
\begin{equation}\label{eqn:baseline}
    U_{b,t} 
    = 
    \mathcal{F}_{b}\left(t, X_t\right)
    = 
    K_b \left(X_t - \Breve{x}_t \right) + \Breve{u}_t, 
\end{equation}
where $K_b \in \mathbb{R}^{m \times n}$ is the feedback gain that is chosen to ensure that $A_\mu+BK_b$ is Hurwitz. 
The reference $(\Breve{x}_t, \Breve{u}_t)$ is given by 
\[
\left( \Breve{x}_t, \Breve{u}_t \right) = \left( \xdisc_{0|k}^\star(\xdisctrue_k), \udisc_{0|k}^\star(\xdisctrue_k) \right),
\]
for each planning interval $t\in[t_k,t_{k+1})$, $t_k=t_0+k\Delta t$, from~\Cref{sec:drmpc}.
Since the nominal system~\eqref{eqn:nom_sys} is uncertainty-free, we set $\bar{U}_t = \mathcal{F}_{b}\left(t, \bar{X}_t\right)$.

Under this baseline law, the nominal closed-loop system is stable since $A_\mu + BK_b$ is Hurwitz, and its moments remain bounded under mild conditions on the diffusion~\cite{tsukamoto2020robust}. However, when applied to the uncertain system, the unknown terms $\Lambda_\mu$ and $\Lambda_\sigma$ can induce unpredictable and unquantifiable divergence between the nominal and true law, as measured by $\Wass_{2p}(\mathbb{X}_t,\bar{\mathbb{X}}_t)$. Hence, as in~\eqref{eqn:aug_cont}, we augment the baseline input with an adaptive component $U_{\mathcal{L}_1,t}$ that can mitigate the effects of the uncertainties.

% Since $A_\mu+BK_b$ is Hurwitz, the moments of the nominal system state $\bar{X}_t$ can be bounded when the diffusion $A_\sigma$ is small enough to be dominated by the closed loop drift $A_\mu+BK_b$~\cite{tsukamoto2020robust}.
% However, if instead of~\eqref{eqn:aug_cont}, one were to set $U_t = \mathcal{F}_b \left(t,X_t\right)$ for the uncertain system~\eqref{eqn:true_sys}, the uncertainties $\Lambda_\mu$ and $\Lambda_\sigma$ can lead to the unpredictable and unquantifiable divergence between the nominal and true laws as captured by $\Wass_{2p}(\mathbb{X}_t,\bar{\mathbb{X}}_t)$. 
% Hence, as in~\eqref{eqn:aug_cont}, we augment the baseline input with the adaptive input $U_{\mathcal{L}_1,t}$ that can mitigate the effects of the uncertainties. 
The adaptive input is given by the \ellonedrac law~\cite{gahlawat2025DRAC}, defined as the output of the following \emph{low pass filter}:
\begin{align}\label{eqn:True:Filter}
    \ULt{t} 
    =  
    - \Boldomega \int_0^t \expo{-\Boldomega(t-\tau)} \Theta_{ad}\Lhat{\tau} d\tau,  
\end{align}
where $\Boldomega \in \mathbb{R}_{>0}$ is the \textbf{filter bandwidth}, $ \Theta_{ad} = \begin{bmatrix}\mathbb{I}_m & 0_{m,n-m}  \end{bmatrix} \bar{B}^{-1} \in \mathbb{R}^{m \times n}$,  $\bar{B} = [B \; B^\perp]\in\mathbb{R}^{n\times n}$ and $B^\perp$ is chosen such that $\operatorname{Im}(B^\perp) = \operatorname{ker}(B^\top)$ and $\operatorname{rank}(\bar{B})=n$.
The \emph{adaptive estimate} $\hat{\Lambda} \in \mathbb{R}^m$ is updated via the following \emph{adaptation law}:
\begin{equation}\label{eqn:True:AdaptationLaw}
    \Lhat{t} 
    =  
    0_m \indicator{[0,\BoldTs)}{t} 
    +
    \lambda_s \br{1 - e^{\lambda_s \BoldTs}}^{-1}
    \sum_{i=1}^{\lfloor \frac{t}{\BoldTs} \rfloor}    
    \Xtildet{i\BoldTs}
    \indicator{[i\BoldTs,(i+1)\BoldTs)}{t},
\end{equation} 
where $\Xtildet{i\BoldTs} = \Xhatt{i\BoldTs} - \Xt{i\BoldTs}$, and $\BoldTs \in \mathbb{R}_{>0}$ is the \textbf{sampling period}.
The parameter $\lambda_s \in \mathbb{R}_{>0}$ contributes to the solution $\Xhatt{t}$ of the \emph{process predictor} given by: 
\begin{equation}\label{eqn:True:ProcessPredictor}
    \Xhatt{t} 
    =
    x_0
    +
    \int_0^t 
        \left(-\lambda_s \mathbb{I}_n \Xtildet{\tau}+ A_\mu \Xt{\tau} + B U_{b,\tau}) \right) d\tau
    +
    \int_0^t 
        \left(B \ULt{\tau} + \Lhat{\tau}\right) d\tau,
\end{equation} 
where $\Xtildet{t} = \Xhatt{t} - \Xt{t}$.
We collectively refer to $\cbr{\Boldomega,\BoldTs,\lambda_s}$ as the \textbf{control parameters} for the \ellonedrac input~\eqref{eqn:True:Filter}~-~\eqref{eqn:True:ProcessPredictor}.

\subsection{Distributional Certificates and Closed-Loop Safety}\label{subsec:DistCerts}

We now formalize the distributional guarantees provided by the input~\eqref{eqn:aug_cont} applied to the uncertain system~\eqref{eqn:true_sys}. 
The following result, shown in~\cite[Thm.~4.1]{gahlawat2025DRAC}, establishes that the low-level robust adaptive input bounds the Wasserstein distance between the uncertain law $\mathbb{X}_t$ and nominal law $\bar{\mathbb{X}}_t$.
%%%%%%%%%%%%%%%%%%%%%%%%%%%%%%%%%%
\begin{theorem}\label{thm:dr_bound}
Suppose~\Cref{ass:true_sys} holds, and let the control input be given by~\eqref{eqn:aug_cont} with $K_b$ in~\eqref{eqn:baseline} ensuring that $A_\mu + B K_b$ is Hurwitz.   
If $\xi_0 \sim \mathcal{P}_{2 p}(\mathbb{R}^{n})$ for some $p \in \mathbb{N}$, then there exists a known constant $\rho_x>0$ such that
$\mathbb X_t \in \mathbb B_{2p}(\bar{\mathbb X}_t,\rho_x),
\; \forall t \ge 0$.
\end{theorem}
%%%%%%%%%%%%%%%%%%%%%%%%%%%%%%%%%%%%%

\Cref{thm:dr_bound} shows that the adaptive controller keeps the true state distribution within a Wasserstein ball around the nominal distribution, thereby certifying the system-side ambiguity set used for planning. While this result characterizes the effects of the epistemic uncertainties on the system state, the environmental state distribution still remains to be estimated from data. 
The following lemma provides a finite-sample bound on the Wasserstein distance between the true environmental state distribution and its empirical estimate.

\begin{lemma}\label{lem:obs_bound}
Let $\bar{\mathbb Z}_t$ denote the empirical distribution constructed from $N$ i.i.d.\ samples of $\mathbb Z_t$. Assume that $\mathbb Z_t$ is light-tailed in the sense that there exist constants $\alpha>2p$ and $A>0$ such that $\mathbb E[\exp(\|Z_t\|^\alpha)]\le A$. Then there exist constants $c_1,c_2>0$ depending only on $\alpha$, $A$, and $n_z$ such that for any confidence level $\delta_z\in(0,1)$, one has
\[
\mathbb{P}_z^{N} \left(
\Wass_{2p} \big(\mathbb Z_t,\bar{\mathbb Z}_t\big)
\le \rho_z(N,\delta_z)
\right)\ge 1-\delta_z,
\]
where $\mathbb{P}_z^{N}$ is the $N$-fold product of the environmental state distribution $\mathbb{P}_z$ and
\[
\rho_z(N,\delta_z)=
\begin{cases}
\left(\frac{\log(c_1/\delta_z)}{c_2 N}\right)^{\min\{\frac{2p}{n_z},\frac{1}{2}\}}
& \text{if } N \ge \frac{\log(c_1/\delta_z)}{c_2}, \\
\left(\frac{\log(c_1/\delta_z)}{c_2 N}\right)^{\frac{2p}{\alpha}}
& \text{otherwise.}
\end{cases}
\]
\end{lemma}

\Cref{lem:obs_bound} follows from concentration inequalities for the Wasserstein distance established in the distributionally robust optimization literature~\cite{fournier2015rate,mohajerin2018data}. Together, \Cref{thm:dr_bound} and \Cref{lem:obs_bound} bound the deviation between the true and nominal system and environmental state distributions thus allowing us to construct ambiguity sets used within the DR-CC~\eqref{eqn:dr_chance_prop}.
\begin{theorem}\label{thm:closed_loop_safety}
    Let the hypotheses of \Cref{thm:dr_bound} hold, and the control input 
    be given by~\eqref{eqn:aug_cont}, and fix $\beta \in (0,1)$ and 
    $\delta_z\in(0,1)$. 
    Further suppose the DR-MPC planner~\eqref{eqn:dr_mpc} is feasible at $t_k$ with ambiguity radius $\rho_y=\rho_x+\rho_z$, where $\rho_x$ and $\rho_z \ge \rho_z(N,\delta_z)$ are given by \Cref{thm:dr_bound} and \Cref{lem:obs_bound}, respectively.
    Then, for each sampling time $t_k$, the following property holds for the true closed-loop system:
    \begin{equation}\label{eqn:prob_guar}
        \mathbb P\!\left(
        X_{t_k} \notin \mathcal O(Z_{t_k})
        \right)\ge 1-\beta,
    \end{equation}
    with probability at least $1-\delta_z$ with respect to the environmental data.
\end{theorem}

The proof is provided in Appendix~\ref{app:proof_cor_safety}.%
\footnote{\Cref{thm:closed_loop_safety} guarantees constraint satisfaction at each step for which~\eqref{eqn:dr_mpc} is feasible, but does not ensure recursive feasibility. If the DR-MPC problem becomes infeasible at some $t_k$, the guarantee no longer applies. Establishing recursive feasibility would require additional terms, such as a distributionally robust terminal set and controller or an infeasibility handling mechanism, which we leave for future work.}

\begin{remark}
The additive structure of $\rho_y$ arises from worst-case coupling arguments and the triangle inequality for the Wasserstein distance, and ensures robustness under independent uncertainties. Tighter ambiguity sets may be obtained by exploiting additional structure or dependence between $X_t$ and $Z_t$, which we leave for future work.
\end{remark}

\Cref{thm:closed_loop_safety} shows that the \ellone certificate and empirical Wasserstein bounds together yield a stagewise safety guarantee at each feasible planning step. \textbf{The adaptive controller bounds the deviation between true and nominal state distributions, 
while the concentration result bounds the environmental distribution deviation}. Together, these determine the ambiguity radii that ensure that the chance constraint~\eqref{eqn:prob_guar} holds at each step where the planner is feasible.

\section{Tractable Reformulation of the DR-MPC Problem}\label{sec:tract_mpc}

The DR-MPC problem in~\eqref{eqn:dr_mpc} contains DR-CC that are not directly amenable to online computation. In this section, we derive a tractable reformulation that replaces these constraints with equivalent CVaR conditions under the nominal distribution. The key idea is to measure safety through the distance to the collision set and to exploit Wasserstein duality to convert the worst-case probability constraint into a deterministic expectation constraint.

\subsection{CVaR Reformulation of the DR-CC}

To obtain a tractable reformulation of~\eqref{eqn:dr_chance_prop}, we consider a fixed time $t$ and invoke the standard strong duality result for Wasserstein distributionally robust optimization~\cite{gao2023distributionally,zhang2025short}.

\begin{lemma}
Let $(\mathcal Y,d_{\mathcal Y})$ be a Polish space, let $\bar{\mathbb Y}_t\in\mathcal P_{2p}(\mathcal Y)$, and let $\ell:\mathcal Y\to\mathbb R$ be Borel measurable with $\ell\in L^1(\bar{\mathbb Y}_t)$. Then
\[
\sup_{\nu \in\mathbb{B}_{2p}(\bar{\mathbb Y}_t,\rho_y)}
\mathbb E_{\nu}[\ell(Y_t)]
=
\inf_{\gamma\ge0}
\left\{
\gamma\rho_y^{2p}
+
\mathbb E_{\bar Y \sim \bar{\mathbb Y}_t}
\!\left[\bar{\ell}_\gamma(\bar{Y})\right]
\right\},
\]
where $\bar{\ell}_\gamma(\bar y)
\coloneqq
\sup_{y\in\mathcal Y}
\left\{
\ell(y)-\gamma d_{\mathcal Y}(y,\bar y)^{2p}
\right\}$.
\end{lemma}

We next specialize this dual representation to our setting, where the inner supremum depends only on the distance between the nominal sample $\bar y$ and the collision set $\mathcal C$. This observation yields an equivalent CVaR characterization of the DR-CC.
Specifically, for a random variable $X \sim \mathbb{P}$ and risk level $\beta \in (0,1)$, the CVaR is defined as
\begin{equation}\label{eqn:cvar_def}
\operatorname{CVaR}^\mathbb{P}_\beta(X) = \inf_{\eta \in \mathbb{R}} \left\{ \eta + \frac{1}{\beta} \mathbb{E}_\mathbb{P}[(X - \eta)_+] \right\}.
\end{equation}
Intuitively, $\operatorname{CVaR}_\beta^\mathbb{P}(X)$ measures the expected value of $X$ in the worst $\beta$-fraction of outcomes, and thus provides a risk-sensitive measure that captures tail behavior beyond standard expectation~\cite{rockafellar2002conditional}. While CVaR is commonly used as a conservative outer approximation of chance constraints (e.g.,~\cite{zymler2013distributionally, xie2021distributionally}), the following result shows that, in our setting, it yields an exact reformulation of the DR-CC under the Wasserstein ambiguity set.

\begin{proposition}
\label{prop:cvar_reform}
Let $(\mathcal Y,d_{\mathcal Y})$ denote the metric space used in the definition of the Wasserstein ambiguity set~\eqref{eqn:amb_set}. Then, for any fixed time $t$, the DR-CC~\eqref{eqn:dr_chance_prop} is equivalent to
\begin{equation}
\operatorname{CVaR}_{\beta}^{\bar{\mathbb Y}_t}
\!\left(-d_{\mathcal Y}(\bar{Y}_t,\mathcal C)^{2p}\right)
\le
-\frac{\rho_y^{2p}}{\beta},
\label{eq:cvar_reform}
\end{equation}
where $ \bar{Y}_t \sim \bar{\mathbb Y}_t$ and $d_{\mathcal Y}(y,\mathcal C) \coloneqq \inf_{c\in\mathcal C} d_{\mathcal Y}(y,c)$ denotes the distance from $y$ to the collision set.
\end{proposition}

The proof is provided in Appendix~\ref{app:proof_cvar_reform}.
Intuitively, condition~\eqref{eq:cvar_reform} requires that the $\beta$-tail expectation of the negative clearance to the collision set remains below zero under the nominal distribution. Thus, the probability of entering the collision set cannot exceed $\beta$ for any distribution within the Wasserstein ambiguity set. Additionally, the reformulation~\eqref{eq:cvar_reform} depends only on the distance to the unsafe set $\mathcal{C}$, and is therefore not restricted to a particular obstacle geometry.

\subsection{Sample-based DR-MPC Problem}
% \todo[inline]{Update discrete distributions}
Since the nominal system dynamics are linear and the process noise is Gaussian,
the predicted state distribution $\xdisclaw{j}{k}$ in the DR-MPC problem~\eqref{eqn:dr_mpc} is Gaussian with mean
$\mu_{j|k}$ and covariance $\Sigma_{j|k}$ obtained from the standard moment propagation equations.
Using~\Cref{prop:cvar_reform}, the DR-MPC problem~\eqref{eqn:dr_mpc} at discrete time $k$ can be reformulated as
\begin{subequations}\label{eqn:sample_mpc}
\begin{align}
\min_{\bar{\mathbf U}_k}
& \sum_{j=0}^{K-1} \ell_j(\mu_{j|k}, \udisc_{j|k}) + \ell_N(\mu_{K|k})\\
\text{s.t.}\;
& \mu_{j+1|k}
=
A_d \mu_{j|k} + B_d\udisc_{j|k},
\qquad j=0,\dots,K-1,
\\
& \udisc_{j|k} \in \mathcal U,
\qquad j=0,\dots,K-1,
\\
& \mu_{j|k} \in \mathcal X,
\qquad j=1,\dots,K,
\\
& \mu_{0|k}=\bm{x},
\\
& \operatorname{CVaR}_{\beta}^{ \ydisclaw{j}{k}}\!
\left(-d_{\mathcal Y}(\ydisc_{j|k},\mathcal C)^{2p}\right)  \le 
-\frac{\rho_y^{2p}}{\beta}, \qquad j=1,\dots,K.
\label{eqn:drmpc_1_cvar}
\end{align}
\end{subequations}

Despite this reformulation, problem~\eqref{eqn:sample_mpc} remains intractable due to the expectation in the CVaR constraint~\eqref{eqn:drmpc_1_cvar}, which is taken with respect to the nominal lifted distribution $\ydisclaw{j}{k}$. To address this issue, we approximate the CVaR constraint using Monte Carlo sampling. 
Recall that the discrete nominal lifted distribution factorizes as $\ydisclaw{j}{k}  = \xdisclaw{j}{k} \otimes \zdisclaw{j}{k}$.
Thus, samples of the lifted variable can be generated by independently sampling the nominal system and environment distributions. In particular, samples of the predicted state are generated as $\hat x_{j|k}^{(s)} \sim \mathcal N(\mu_{j|k},\Sigma_{j|k}), \,s=1,\dots,N$, while samples $\hat z_{j|k}^{(s)}$ are drawn from the known nominal environment distribution $\zdisclaw{j}{k}$. The resulting lifted samples are
$\hat y_{j|k}^{(s)} = \big(\hat x_{j|k}^{(s)}, \hat z_{j|k}^{(s)} \big), \, s=1,\dots,N$.

Using these samples, the CVaR term in~\eqref{eqn:drmpc_1_cvar} is approximated via its sample-average representation
\begin{equation*}\label{eqn:cvar_saa}
\operatorname{CVaR}_{\beta}^{ \ydisclaw{j}{k} }
\left(-d_{\mathcal Y}(\bar y_{j|k},\mathcal C)^{2p}\right)\approx
\inf_{\eta\in\mathbb R}
\Bigg[
\eta + \frac{1}{\beta N}\sum_{s=1}^N
\big(-d_{\mathcal Y}(\hat y_{j|k}^{(s)},\mathcal C)^{2p}-\eta\big)_+
\Bigg].
\end{equation*}

The resulting formulation replaces the infinite-dimensional DR-CC with a finite set of constraints based on sampled nominal realizations of the environment states. As the number of samples $N$ increases, the sample-average approximation converges to the true CVaR constraint.\footnote{For several common environmental geometries, the distance $d_{\mathcal Y}((x,z),\mathcal C)$ admits closed-form expressions, yielding tractable constraints. In particular, circular and polytopic unsafe sets lead to efficiently computable formulations; see Appendix~\ref{app:obstacles}.}
The resulting \frameworkname \ framework is summarized in \Cref{alg:dr_l1}.
% In the next section, we summarize the resulting distributionally robust planning and control architecture in the form of an implementable algorithm.

% \subsection{Algorithm}

% The overall planning and control architecture is summarized in \Cref{alg:dr_l1}.
% The proposed framework operates on two time scales. The high-level DR-MPC planner is updated every $\Delta t$, while the low-level $\mathcal L_1$ adaptive controller runs at the faster sampling period $T_s$, typically with $T_s \ll \Delta t$. At each planning step, the planner computes a nominal trajectory by solving the sample-based DR-MPC problem~\eqref{eqn:sample_mpc}, where the CVaR
% constraint is approximated using the sample-average formulation~\eqref{eqn:cvar_saa}. The first control input of this trajectory is then tracked over the interval $[t_k,t_{k+1})$ using zero-order hold and augmented with the adaptive compensation term.

\begin{algorithm}[t]
\caption{\frameworkname}
\label{alg:dr_l1}
\begin{algorithmic}[1]
\State \textbf{Input:} horizon $K$, samples $N$, ambiguity radii $\rho_x, \rho_z$, planning period $\Delta t$, adaptation period $T_s$, filter bandwidth $\omega$, adaptation gain $\lambda_s$, control gain $K_b$
\For{$k=0,1,2,\dots$}
    \State $t_k \gets k\Delta t$; measure $X_{t_k}$
    \State Sample $\{\hat z_{j|k}^{(s)}\}_{j=1:K}^{s=1:N}$ from $\zdisclaw{j}{k}$
    \State Solve DR-MPC~\eqref{eqn:sample_mpc} with $\bm{x}=X_{t_k}$ to obtain $(\mathbf{\udisc}_k^\star(\bm{x}), \mathbf{\xdisc}_k^\star(\bm{x}))$
    \State Set $\Breve{u}_t=\udisc_{0|k}^\star(\bm{x})$, $\;\Breve{x}_t=\xdisc_{0|k}^\star(\bm{x})$  
    \For{$t=t_k: T_s: t_{k+1}-T_s$}
        \State Measure $X_t$
        \State $U_t \gets \mathcal{F}_b(t,X_t) + U_{\mathcal L_1,t}$
        \State Update $\mathcal L_1$ predictor, adaptation, and filter
    \EndFor
\EndFor
\end{algorithmic}
\end{algorithm}

\section{Experiments}\label{sec:exp}
We evaluate the proposed framework through simulations of trajectory tracking with dynamic obstacle avoidance, where the true system is subject to nonlinear drift and diffusion uncertainties
% that are not captured by the planning model. 
We compare three methods: \emph{DRP-Baseline (Nominal)}, which plans for the uncertainty-free nominal system operating in an uncertain environment and applies the baseline controller; \emph{DRP-Baseline (True Dynamics)}, which applies the same controller directly to the uncertain system without the \ellone augmentation; and \emph{\frameworkname \ (Proposed)}, which augments the planner with an \ellone-adaptive input to compensate for model uncertainties.\footnote{Additional experimental results, videos, and code are available at the \href{https://github.com/astghikhakobyan/DRP-L1AC}{project's codebase.}}

\subsection{Experiment Setup}
We consider a robot operating in a 2D environment with state $x_t = [p_{x,t}, p_{y,t}, v_{x,t}, v_{y,t}]^\top$, representing planar position and velocity, and control input $u_t = [a_{x,t}, a_{y,t}]^\top$, representing accelerations in the $x$- and $y$-directions. 
The nominal dynamics are modeled as a stochastic double integrator, while the true system includes nonlinear drift and diffusion uncertainties. 
The environment consists of multiple moving obstacles, each following a constant-velocity profile perturbed by additive Gaussian noise (standard deviation $0.5$).\footnote{The extension of \frameworkname \ to multiple environmental components (e.g., obstacles) is readily obtained by introducing variables $\{Z_t^{(i)}\}_{i=1}^{n_o}$, where $n_o$ denotes the number of components, and enforcing one DR-CC of the form \eqref{eqn:dr_chance_prop} for each component.}

The DR-MPC planner tracks a desired trajectory with prediction horizon $K=25$, risk level $\beta=0.05$, and $N=40$ samples for the CVaR approximation. The \ellone-adaptive controller uses $\omega=50$, $T_s=10^{-3}$, and $\lambda_s=100$. We evaluate four reference trajectories (\emph{figure-eight}, \emph{circle}, \emph{Lissajous}, and \emph{spiral}) over 100 Monte Carlo runs. The ambiguity radius for each case is shown in~\Cref{tab:controller_results}, selected using the certified bound of \Cref{thm:dr_bound} and~\Cref{lem:obs_bound} as an initial design point and then tuned to balance conservatism and empirical performance.
% \footnote{The certified construction yields the conservative radius $\rho_x + \rho_z$, which we use as an initial design point. Therefore, we employ a smaller tuned radius $\rho_y$ to balance robustness and empirical performance.}

\subsection{Results}
\paragraph{Quantitative Results.}
\Cref{tab:controller_results} summarizes the tracking cost (for successful runs), Wasserstein distance, and failure probability across 100 Monte Carlo runs, where failure includes collisions, loss of feasibility, or numerical instability (e.g., state divergence).

While the distributionally robust planner accounts for environmental uncertainty, it does not compensate for the system's endogenous uncertainties. 
As a result, applying the baseline controller directly to the true system leads to severe degradation in closed-loop performance, including large increases in tracking cost and high failure rates (up to $86\%$ in the circle scenario), indicating a lack of robustness to unmodeled dynamics.
In contrast, augmenting DR-MPC with the \ellone-adaptive controller significantly improves reliability. 
The proposed method eliminates failures in all four scenarios, while maintaining tracking performance.

% Importantly, the average empirical Wasserstein distance $\Wass_{\mathrm{mean}}$ between the true and nominal state distributions remains consistently below the certified bound, indicating consistency with the theoretical distributional certificate. This suggests that the true system evolution remains close to the nominal distribution used for planning, supporting the validity of the DR-CCs at each feasible planning step.

\begin{table}[t]
\centering
\caption{Performance comparison over 100 runs (mean $\pm$ std).}
\label{tab:controller_results}
\begin{tabular}{llcccc}
\toprule
\textbf{Reference} & \textbf{Method} & \textbf{$J_{\text{track}}$} &  \textbf{$P_{\text{fail}}(\%)$} & \textbf{$\rho_{\text{emp}}$}  \\
\midrule
\multirow{3}{*}{Figure-eight}
 & DRP-Baseline (Nom.) 
 & $7554.79 \pm 1228.75$ 
 & $0$ 
 & \multirow{3}{*}{$0.04$}\\
 & DRP-Baseline (True Dyn.) 
 & $12102.00 \pm 6865.03$ 
 & $24$ &\\
 & \frameworkname \ (Proposed)
 & $8198.77 \pm 1690.13$ 
 & $0$ &\\

\midrule
\multirow{3}{*}{Circle}
 & DRP-Baseline (Nom.) 
 & $1883.42 \pm 235.43$  
 & $0$ 
 & \multirow{3}{*}{$0.02$}\\
 & DRP-Baseline (True Dyn.) 
 & $3.68{\times}10^{5} \pm 3.65{\times}10^{5}$ 
 & $86$ &\\
 & \frameworkname \ (Proposed)
 & $2019.25 \pm 246.64$ 
 & $0$ &\\

\midrule
\multirow{3}{*}{Lissajous}
 & DRP-Baseline (Nom.) 
 & $566.76 \pm 93.62$ 
 & $0$ 
 & \multirow{3}{*}{$0.04$}\\
 & DRP-Baseline (True Dyn.) 
 & $3090.88 \pm 2091.68$ 
 & $8$ &\\
 & \frameworkname 
 & $681.49 \pm 79.25$ 
 & $0$ &\\

\midrule
\multirow{3}{*}{Spiral}
 & DRP-Baseline (Nom.) 
 & $5657.09 \pm 3152.24$ 
 & $2$ 
 & \multirow{3}{*}{$0.04$}\\
 & DRP-Baseline (True Dyn.) 
 & $4248.69 \pm 1776.33$ 
 & $44$ &\\
 & \frameworkname \ (Proposed)
 & $6236.01 \pm 5214.38$ 
 & $1$ &\\

\bottomrule
\end{tabular}
\end{table}

\paragraph{Qualitative Results.} 
Representative trajectories for the figure-eight scenario are illustrated~\Cref{fig:figure8}. Although all methods rely on the same DR-MPC planner, their closed-loop behavior differs significantly under system's epistemic uncertainties. 
At $t = 1.58\,\text{s}$, the true system without adaptive augmentation begins to deviate noticeably from the reference trajectory due to unmodeled dynamics. 
Furthermore, at $t=5.30 \, \text{s}$, the system fails to anticipate the uncertainty in the cluttered region, leading to a collision with the polygon.

By contrast, the proposed \frameworkname \ approach closely tracks the nominal trajectory throughout execution. This behavior is consistent with the theoretical guarantees, which ensure that the adaptive controller keeps the true state distribution close to the nominal one. 

Overall, the results highlight the complementary roles of the two components: DR-MPC provides robustness to environmental uncertainty at the planning stage, while the \ellone-adaptive controller compensates for model uncertainty online, ensuring reliable closed-loop performance and maintaining the validity of the distributional safety guarantees at each feasible planning step.

\section{Conclusions and Future Work}\label{sec:conc}

This paper proposed a hierarchical framework for safe control under model and environmental uncertainties by integrating \ellone-AC with DR-MPC. The \ellone controller provides online certificates that bound the deviation between nominal and true state distributions, enabling the construction of ambiguity sets for DR-MPC. We showed that the feasibility of the resulting planning problem implies stagewise safety of the true system. Numerical results demonstrated effective performance under simultaneous uncertainties.

Future work will consider systems with nominal nonlinear dynamics with unmatched uncertainties, as well as the recursive feasibility of DR-MPC.
A natural extension is the integration of data-driven components (e.g., learned dynamics, controllers, and perception). 
The distributional structure of the proposed framework enables its principled incorporation through compatible guarantees (e.g., empirical risk on training distributions).

\begin{figure*}
    \centering
    \includegraphics[width=\linewidth]{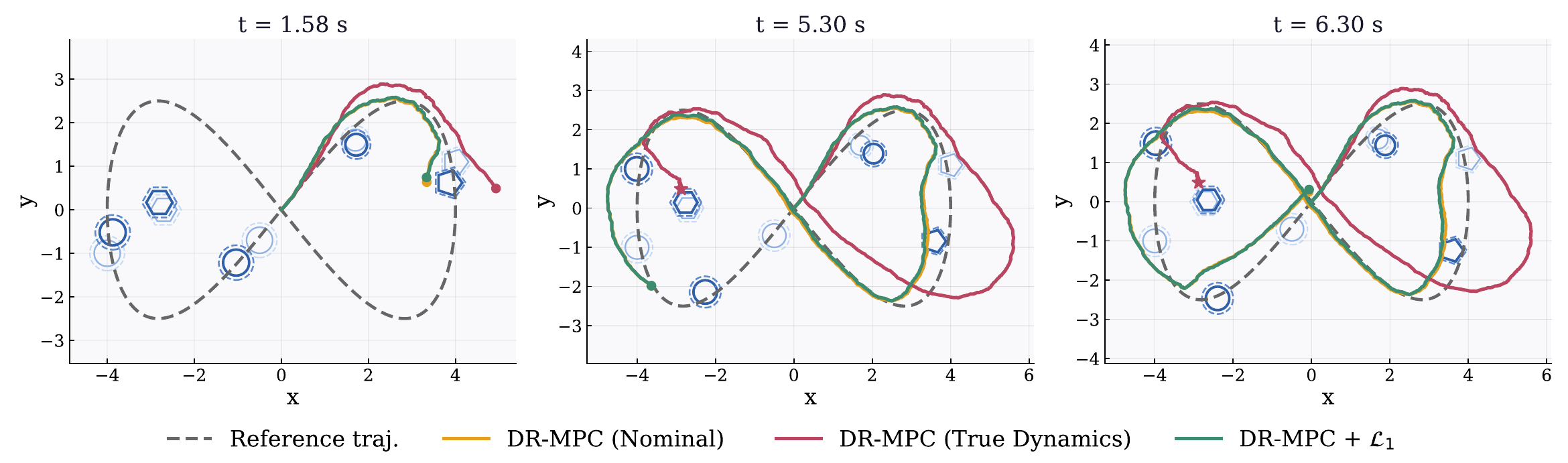}
    \caption{Closed-loop trajectories for the figure-eight scenario at different times. Without adaptation, the controller deviates from the reference due to model mismatch, whereas the \ellone-augmented controller closely tracks the nominal trajectory and maintains safe obstacle avoidance.}
    \label{fig:figure8}
\end{figure*}

%%%%%%%%%%%%%%%%%%%%%%%%%%%%%%%%%%%%%%%%%%%%%%%%%%%%%%%%%%%%%%%%%%%%%%%%%%%%%%%%
\appendix

\section{Proof of~\Cref{thm:closed_loop_safety}}
\label{app:proof_cor_safety}
\begin{proof}
Fix a sampling time $t_k$. Define the lifted random variable
$Y_{t_k} := (X_{t_k},Z_{t_k})$, with law $\mathbb Y_{t_k} = \mathcal L(Y_{t_k})$, and let $\bar Y_{t_k} := (\bar X_{t_k},\bar Z_{t_k})$
denote its nominal counterpart with law $\bar{\mathbb Y}_{t_k} = \bar{\mathbb X}_{t_k}\otimes \bar{\mathbb Z}_{t_k}$.

By \Cref{thm:dr_bound}, we have 
\[
\Wass_{2p}(\mathbb X_{t_k},\bar{\mathbb X}_{t_k}) \le \rho_x.
\]
Moreover, since $\rho_z \ge \rho_z(N,\delta_z)$, \Cref{lem:obs_bound} implies that 
\[
\Wass_{2p}(\mathbb Z_{t_k},\bar{\mathbb Z}_{t_k}) \le \rho_z
\]
holds with probability at least $1-\delta_z$ with respect to the environmental data.

Conditioned on this event, let $\pi_x\in\Pi(\mathbb X_{t_k},\bar{\mathbb X}_{t_k})$ and $\pi_z\in\Pi(\mathbb Z_{t_k},\bar{\mathbb Z}_{t_k})$ be couplings.
Since the system and environmental states are independent, the joint distributions factorize as $\mathbb Y_{t_k} = \mathbb X_{t_k}\otimes \mathbb Z_{t_k}$ and $\bar{\mathbb Y}_{t_k} = \bar{\mathbb X}_{t_k}\otimes \bar{\mathbb Z}_{t_k}$. Therefore, the product measure $\pi := \pi_x \otimes \pi_z$ is a coupling between $\mathbb Y_{t_k}$ and $\bar{\mathbb Y}_{t_k}$.
By definition of the Wasserstein distance,
\[
\Wass_{2p}(\mathbb Y_{t_k},\bar{\mathbb Y}_{t_k})
\le
\left(
\int \|(x,z)-(\bar x,\bar z)\|^{2p}\,d\pi
\right)^{1/(2p)}.
\]
Using the triangle inequality and Minkowski's inequality, we have
\begin{align*}
 \Wass_{2p}(\mathbb Y_{t_k},\bar{\mathbb Y}_{t_k})
& \le
\left(
\int (\|x-\bar x\|+\|z-\bar z\|)^{2p}\,d\pi
\right)^{1/(2p)} \\
& \le
\Wass_{2p}(\mathbb X_{t_k},\bar{\mathbb X}_{t_k})
+
\Wass_{2p}(\mathbb Z_{t_k},\bar{\mathbb Z}_{t_k}) \\
&\le \rho_x + \rho_z.
\end{align*}
Hence, $\mathbb Y_{t_k} \in \mathbb B_{2p}\bigl(\bar{\mathbb Y}_{t_k},\,\rho_x+\rho_z\bigr)$ with probability at least $1-\delta_z$ with respect to the environmental data.

Since the DR-MPC planner enforces the DR-CC~\eqref{eqn:drmpc_safe} for all distributions in the ambiguity set $\mathbb B_{2p}(\bar{\mathbb Y}_{t_k},\rho_y)$ with $\rho_y=\rho_x+\rho_z$, it follows that 
\[
\mathbb P(Y_{t_k}\in \mathcal C)\le \beta
\]
with probability at least $1-\delta_z$.
Finally, since 
\[
Y_{t_k}\in \mathcal C
\Longleftrightarrow
X_{t_k} \in \mathcal O(Z_{t_k}),
\]
we conclude that $\mathbb P\!\left(
X_{t_k} \notin \mathcal O(Z_{t_k})
\right)\ge 1-\beta$, which completes the proof.
\end{proof}

\section{Proof of~\Cref{prop:cvar_reform}}
\label{app:proof_cvar_reform}

\begin{proof}
Let $\ell(y)=\indicator{\mathcal C}{y}$. Since $\mathcal C$ is measurable, the function $\ell$ is Borel measurable, and since $0\le \ell\le 1$, we also have $\ell\in L^1(\bar{\mathbb Y}_t)$. Therefore, the strong duality lemma applies and yields
\[
\sup_{\nu \in\mathbb B_{2p}(\bar{\mathbb Y}_t,\rho_y)}
\mathbb P_{\nu}(Y_t\in\mathcal C)
=
\inf_{\gamma\ge 0}
\left\{
\gamma\rho_y^{2p}
+
\mathbb E_{\bar Y_t\sim\bar{\mathbb Y}_t}
\big[\bar\ell_\gamma(\bar Y_t)\big]
\right\},
\]
where $\bar\ell_\gamma(\bar y)
=
\sup_{y\in\mathcal Y}
\left\{
\mathbf 1_{\mathcal C}(y)-\gamma d_{\mathcal Y}(y,\bar y)^{2p}
\right\}$.

Now fix $\bar y\in\mathcal Y$. If $y\in\mathcal C$, the inner expression equals
$1-\gamma d_{\mathcal Y}(y,\bar y)^{2p}$, whose supremum over $y\in\mathcal C$ is
$1-\gamma d_{\mathcal Y}(\bar y,\mathcal C)^{2p}$.
If $y\notin\mathcal C$, the indicator term vanishes, and the expression is at most $0$. Hence
\[
\bar\ell_\gamma(\bar y)
=
\bigl(1-\gamma d_{\mathcal Y}(\bar y,\mathcal C)^{2p}\bigr)_+.
\]

Substituting this into the dual formulation gives
\[
\sup_{\nu\in\mathbb B_{2p}(\bar{\mathbb Y}_t,\rho_y)}
\mathbb P_{\nu}(Y_t\in\mathcal C) =
\inf_{\gamma\ge 0}
\left\{
\gamma\rho_y^{2p}
+
\mathbb E_{\bar Y_t\sim\bar{\mathbb Y}_t}
\left[
\bigl(
1-\gamma d_{\mathcal Y}(\bar Y_t,\mathcal C)^{2p}
\bigr)_+
\right]
\right\}.
\]
Therefore, the DR-CC~\eqref{eqn:dr_chance_prop} holds if and only if
\begin{equation}
\inf_{\gamma\ge 0}
\left\{
\gamma\rho_y^{2p}
+
\mathbb E_{\bar Y_t\sim\bar{\mathbb Y}_t}
\left[
\bigl(
1-\gamma d_{\mathcal Y}(\bar Y_t,\mathcal C)^{2p}
\bigr)_+
\right]
\right\}
\le \beta .
\label{eq:proof_aux_dual}
\end{equation}

Since $\beta\in(0,1)$, the optimizer in~\eqref{eq:proof_aux_dual} cannot occur at $\gamma=0$, because the objective value at $\gamma=0$ is $1>\beta$. Thus, we may restrict attention to $\gamma>0$. Introduce the change of variables $\eta=-1/\gamma$. Then
\[
\bigl(
-d_{\mathcal Y}(\bar Y_t,\mathcal C)^{2p}-\eta
\bigr)_+
=
\frac{1}{\gamma}
\bigl(
1-\gamma d_{\mathcal Y}(\bar Y_t,\mathcal C)^{2p}
\bigr)_+.
\]
Therefore,
\[
\eta
+
\frac{1}{\beta}
\mathbb E_{\bar Y_t\sim\bar{\mathbb Y}_t}
\left[
\bigl(
-d_{\mathcal Y}(\bar Y_t,\mathcal C)^{2p}-\eta
\bigr)_+
\right]
\le
-\frac{\rho_y^{2p}}{\beta},
\]
if and only if 
\[
\gamma\rho_y^{2p}
+
\mathbb E_{\bar Y_t\sim\bar{\mathbb Y}_t}
\left[
\bigl(
1-\gamma d_{\mathcal Y}(\bar Y_t,\mathcal C)^{2p}
\bigr)_+
\right]
\le \beta.
\]

Using the standard representation~\eqref{eqn:cvar_def} and taking the infimum over $\eta\in\mathbb{R}$ (equivalently, over $\gamma>0$) yields~\eqref{eq:cvar_reform}.
\end{proof}

\section{Tractable Distance Expressions for Translated Unsafe Regions}
\label{app:obstacles}

In this appendix, we derive tractable expressions for the distance 
$d_{\mathcal Y}((x,z),\mathcal C)$ appearing in the CVaR reformulation of the distributionally robust chance constraint, when the environment induces unsafe regions via translations of a fixed shape.

We consider the case where the environment variable $z\in\mathbb{R}^{n_z}$ represents the position of the obstacle, while the system state $x\in\mathbb{R}^n$ may include additional components (e.g., velocity). Let $C \in \mathbb{R}^{n_z \times n}$ be a matrix that extracts the position components of the state, so that $Cx$ represents the position of the system.

\begin{lemma}
\label{lem:lifted_distance_translated}
Let $(\mathcal{Y}, d_{\mathcal Y})$ denote the lifted state space equipped with the Euclidean metric 
\[
d_{\mathcal Y}((x,z),(x',z')) = \|(x,z) - (x',z')\|.
\]
Let $S \subset \mathbb{R}^{n_z}$ be a nonempty closed set and define the translated unsafe region
\[
\mathcal O(z) = \{x \in \mathbb{R}^n : Cx - z \in S\}, \qquad z \in \mathbb{R}^{n_z}.
\]
The corresponding lifted unsafe set is 
\[
\mathcal C
=
\{(x,z)\in\mathbb{R}^n\times\mathbb{R}^{n_z} : Cx - z \in S\}.
\]
Then, for any $(x,z)\in\mathcal Y = \mathbb{R}^n\times\mathbb{R}^{n_z}$,
\[
d_{\mathcal Y}\big((x,z),\mathcal C\big)
=
\frac{1}{\sqrt 2}\operatorname{dist}(Cx-z,S),
\]
where $\operatorname{dist}(y,S) := \inf_{\tilde y \in S}\|y-\tilde y\|$.
\end{lemma}

\begin{proof}
By definition,
\[
d_{\mathcal Y}\big((x,z),\mathcal C\big)
=
\inf_{\substack{(x',z')\\ Cx'-z'\in S}}
\sqrt{\|x-x'\|^2+\|z-z'\|^2}.
\]

Let $y := Cx - z$ and $\tilde y := Cx' - z'$. Then
$\tilde y \in S$ and
\[
y - \tilde y = C(x-x') - (z-z').
\]
Define $a := x - x'$ and $b := z - z'$. Then $C a - b = y - \tilde y$.
Fix $d := y - \tilde y$. The inner problem becomes
\[
\inf_{\substack{a\in\mathbb{R}^n,\; b\in\mathbb{R}^{n_z}\\ C a - b = d}}
\sqrt{\|a\|^2 + \|b\|^2}.
\]

The optimal solution is obtained by balancing $a$ and $b$, yielding the minimum value $\|d\|/\sqrt{2}$.

Therefore,
\[
d_{\mathcal Y}\big((x,z),\mathcal C\big)
=
\frac{1}{\sqrt 2}\inf_{\tilde y\in S}\|y-\tilde y\|
=
\frac{1}{\sqrt 2}\operatorname{dist}(Cx-z,S).
\]
\end{proof}

\medskip

This result highlights that the lifted distance depends only on the relative position $Cx - z$, reducing the problem to computing a standard Euclidean distance in $\mathbb{R}^{n_z}$ scaled by $1/\sqrt{2}$.

We now specialize this result to common obstacle geometries.

\subsection{Circular Unsafe Region}
\label{app:obstacles_circular}

Consider a circular unsafe region with uncertain center $z \in \mathbb{R}^{n_z}$ and radius $r>0$, defined by
\[
\mathcal O(z)
=
\{x \in \mathbb{R}^n : \|Cx - z\| \le r\}.
\]
This corresponds to the base set $S = \{y\in\mathbb{R}^{n_z} : \|y\|\le r\}$.
By \Cref{lem:lifted_distance_translated},
\[
d_{\mathcal Y}\big((x,z),\mathcal C\big)
=
\frac{1}{\sqrt 2}\left(\|Cx - z\| - r\right)_+.
\]

Thus, for sampled lifted variables,
\[
d_{\mathcal Y}\bigl(\hat y_{j|k}^{(s)},\mathcal C\bigr)
=
\frac{1}{\sqrt 2}
\left(
\bigl\|C\hat x_{j|k}^{(s)}-\hat z_{j|k}^{(s)}\bigr\| - r
\right)_+.
\]

\subsection{Polytopic Unsafe Region}
\label{app:obstacles_polytope}

Consider a translated polytopic unsafe region with uncertain position $z \in \mathbb{R}^{n_z}$ and fixed shape $P := \{y \in \mathbb{R}^{n_z} : A y \le b \}$,
where $A \in \mathbb{R}^{m\times n_z}$ and $b \in \mathbb{R}^m$. The unsafe region is given by
\[
\mathcal O(z)=\{x \in \mathbb{R}^n : Cx - z \in P\}.
\]

Applying \Cref{lem:lifted_distance_translated} yields
\[
d_{\mathcal Y}\big((x,z),\mathcal C\big)
=
\frac{1}{\sqrt 2}\operatorname{dist}(Cx-z,P).
\]
Equivalently,
\[
d_{\mathcal Y}\big((x,z),\mathcal C\big)
=
\frac{1}{\sqrt 2}
\min_{\substack{\tilde y\in\mathbb{R}^{n_z}\\ A\tilde y\le b}}
\|Cx - z - \tilde y\|.
\]

Thus, for sampled lifted variables,
\[
d_{\mathcal Y}\big(\hat y_{j|k}^{(s)},\mathcal C\big)
=
\frac{1}{\sqrt 2}
\min_{\substack{\tilde y\in\mathbb{R}^{n_z}\\ A\tilde y\le b}}
\bigl\|
C\hat x_{j|k}^{(s)}-\hat z_{j|k}^{(s)}-\tilde y
\bigr\|.
\]

In both cases, the distance term admits a convex representation and can be incorporated into the sample-based CVaR constraint using standard epigraph reformulations.

\bibliographystyle{IEEEtran}
\bibliography{references}

\end{document}